\documentclass{amsart}
\usepackage{setspace}
\usepackage{amsmath}
\usepackage{amssymb}
\usepackage{fullpage}
\usepackage{graphicx}
\usepackage{xypic}
\usepackage{algorithm}
\usepackage{algorithmic}

\usepackage{color}

\newcommand{\GRS}{\textnormal{GRS}}
\newcommand{\F}{\mathbb{F}}
\newcommand{\T}{\textnormal{T}}
\newcommand{\tr}{\textnormal{tr}}
\newcommand{\Ev}{\textnormal{Ev}}

\newtheorem{theorem}{Theorem}[section]
\newtheorem{definition}[theorem]{Definition}
\newtheorem{proposition}[theorem]{Proposition}
\newtheorem{lemma}[theorem]{Lemma}

\title{The Dimension of Subcode-Subfields of Shortened Generalized Reed Solomon Codes}
\author{Fernando Hernando \and Kyle Marshall\and Michael E. O'Sullivan}
\thanks{The work of K. Marshall and M. O'Sullivan is supported by the National Science Foundation under Grant No. CCF-0916492.} 

\begin{document}

\begin{abstract}
Reed-Solomon (RS) codes are among the most ubiquitous codes due to  their good parameters as well as efficient encoding and decoding procedures. However, RS codes suffer from having a fixed length. In many applications where the length is static, the appropriate length can be obtained by an RS code by shortening or puncturing. Generalized Reed-Solomon (GRS) codes are a generalization of RS codes, whose subfield-subcodes are extensively studied. In this paper we show that a particular class of GRS codes produces many subfield-subcodes with large dimension. An algorithm for searching through the codes is presented as well as a list of new codes obtained from this method.
\end{abstract}

\maketitle

\section{Introduction}

Generalized Reed-Solomon ($\GRS$) codes are  among the 
most widely studied codes because of their useful algebraic properties and their efficient encoding and decoding algorithms. Subfield-subcodes of $\GRS$ codes are known as alternant codes, a class which includes popular codes such as  BCH codes and the classical Goppa codes. In general, the minimum distance and the dimension of a subfield-subcode (SFSC) are not trivial to obtain. However, since $\GRS$ codes are maximum distance separable, the minimum distance of the SFSC is at least bounded by the minimum distance of the original code. This property can be used to construct good subfield-subcodes (SFSC) from $\GRS$ codes. 

Although there is no general formula for the dimension or the minimum distance of alternant codes, many good bounds and special cases have been obtained. Delsarte \cite{Delsarte1} studied  the subfield-subcodes of generalized Reed-Solomon codes and their duals, which are trace codes. Stichtenoth improved Delsarte's  lower bound on dimension in \cite{Stichtenoth1} and Shibuya et al gave a better lower bound \cite{Tomoharu- Ryutaroh-Kohichi}.   Closed formulae are obtained for the true dimension of certain alternant codes in \cite{Dianwu},  for  subfield-subcodes of toric codes in \cite{Hernando1} and for a family of  Goppa codes in \cite{Veron1}.

In this paper, we establish a formula which provides a lower bound for the dimension of a specific class of alternant codes. We are motivated by the following idea of Roseiro et al  \cite{Roseiro1}: use the kernel of the associated trace map to compute the dimension of the subfield-subcode. Veron also use this idea \cite{Veron1} to compute the true dimension of some binary Goppa codes. 

In Section~2 we introduce the tools neccessary for the rest of the
paper. In particular, we provide an interpretation for Generalized
Reed-Solomon codes as evaluation codes. In Section~3,  we introduce and deduce important properties of a map $T$ extending the trace map.  The interaction between $T$ and cyclotomic polynomials  helps us to understand  punctured Reed-Solomon codes. In Section~4 we give a lower bound for the dimension of  Subfield-Subcodes of certain shortened GRS codes. In section 5 we propose two algorithms that offer  different search patterns through the same set of codes. We found $98$ codes improving the parameters for best-known codes using this algorithms.

\section{Background}

Throughout this article $p$ is a prime power, $m$ is a postive integer and $1 \leq n \leq p^m-1$.  We also set $N=p^m-1$.

\begin{definition}\label{d:grscode}  Let $\alpha = (\alpha_1,\alpha_2,...,\alpha_{n})$ be a vector of distinct non-zero elements of $\F_{p^{m}}$, and let $v = (v_1,v_2,...,v_{n})\in (\F_{p^m}^*)^n$. The Generalized Reed-Solomon code over $\F_{p^m}$ with length $n$ and dimension $k \leq n$ is defined to be $$\GRS_k(\alpha,v) = \{\Ev_{\alpha,v}(f) \colon \deg f < k\}$$ where $\Ev_{\alpha,v}(f) = (f(\alpha_1)v_1, f(\alpha_2)v_2, ..., f(\alpha_{n})v_{n})$. \end{definition}
                \vskip 6pt

The vector $\alpha$ specifies the evaluation points of the code, and $v$ is a vector we will call the twist vector. A classical narrow-sense Reed-Solomon code occurs when $n = p^m-1$ and $\alpha = (1,\eta, \eta^2, ..., \eta^{n-1})$, where $\eta$ is primitive and $v = \overline{1} = (1,1,...,1)$. The family $\GRS$ codes is closed under taking duals. In particular, $$\GRS_k(\alpha,v)^\perp = \GRS_{n-k}(\alpha, u)$$ where $u = ( u_1,u_2...,u_n)$ and $u_{i}^{-1} =v_i \prod_{j\neq i}(\alpha_i - \alpha_j)$ \cite[ Theorem 5.3.3]{HuffmanPless}.

\begin{definition}\label{d:sfsc} Let $\mathcal{C}$ be a code over $\F_{p^m}$. The subfield-subcode of $\mathcal{C}$, $\mathcal{C}|_{\F_p}$, is given by the codewords $c\in\mathcal{C}$ whose coordinates all belong to the subfield $\F_p$. That is, if $n$ is the length of $\mathcal{C}$, then $$\mathcal{C}|_{\F_p} = \mathcal{C}\cap\F_{p}^n = \{c\in\mathcal{C}\colon c\in\F_{p}^n\}.$$ \end{definition}
                \vskip 6pt

 GRS codes form a large class of maximum distance separable codes. Since the subfield-subcode will always have minimum distance at least as great as that of the original code, a lower bound for the minimum distance of a subfield-subcode can be taken as the minimum distance of the original code. A lower bound for the dimension can be obtained by expanding the elements of the check matrix as $m$-dimensional vectors over $\F_p$, giving $$\dim_{\F_p}(\mathcal{C}\cap \F_p^n) \geq n - m(n-k).$$ This bound is very loose, and better bounds exist for specific classes of codes \cite{Stichtenoth1}. In the case of $\GRS$ codes, it is not known which twist vectors give rise to GRS codes that have subfield-subcodes with large dimension. A large class of subfield-subcodes of $\GRS$ codes with good parameters are the Goppa codes \cite[ Section13.5.1]{HuffmanPless}.

The following result is a straightforward application of the Chinese Remainder Theorem.

\begin{lemma}\label{l:iso} Let $R = \F_{p^m}[x]/\langle x^N - 1\rangle$ and $\alpha = (1,\eta,...,\eta^{N-1})$ where $\eta$ is a primitive element in $\F_{p^m}$. Then, $R$ is isomorphic to $\F_{p^m}^N$ under the map $\Ev$, where $$\Ev \colon R\rightarrow \F_{p^m}^N,\ \ \ \ \ f\mapsto (f(1), f(\eta), ..., f(\eta^{N-1})).$$ $\Ev^{-1}$is given by Lagrange Interpolation. 

\end{lemma}
                \vskip 6pt

We extend the definition of degree to elements of $R$ in the obvious manner, by choosing the minimal degree representative in $\F_{p^m}[x]$.

\begin{proposition}\label{p:newgrs} Let $\alpha = (\alpha_1, \alpha_2, ..., \alpha_{n})$ be an $n$-tuple of distinct elements of $\F_{p^m}^*$ and $v \in (\F_{p^m}^*)^n$. Then we can define a unique $g\in R$ such that $$g(x) = \left\{\begin{array}{ll} v_i & x = \alpha_i \\ 0 & x\in \F_{p^m}^*\setminus \{\alpha_i\}_{i=1}^{n} \end{array}\right.$$ and $$\GRS_k(\alpha, v) = \{\Ev_{Z_g^c, \bar{1}}(f\cdot g)\colon \deg f < k\}$$ where $$Z_g = \{x\in\F_{p^m}^* \colon g(x) = 0\}$$ and $Z_g^c$ denotes the complement of $Z_g$ in $\F_{p^m}^*$. \end{proposition}
                \vskip 6pt

This allows us to define the $\GRS$ code in an equivalent way that  is more concise and will prove useful in deriving bounds. To emphasize the relationship between a $\GRS$ code and $g$, we state the following definition.

\begin{definition}\label{d:newgrs} Let $g$ be as in Proposition \ref{p:newgrs} and $n = |Z_g^c|$. Then, the code $$\GRS_k\langle g\rangle = \{\Ev_{Z_g^c, \bar{1}}(f\cdot g)\colon \deg f < k\}$$ will be called the Generalized Reed-Solomon code with twist polynomial $g$. \end{definition}
                \vskip 6pt

A narrow sense Reed-Solomon code is $\GRS_k\langle 1\rangle$ and its dual is  $\GRS_{N-k}\langle x\rangle$.

 The general result of Delsarte is our starting point for creating codes with good parameters.

\begin{theorem}[Delsarte]\label{t:delsarte} Let $\mathcal{C}$ be a linear code of length $n$ over $\mathbb{F}_{p^m}$. Then, $$(\mathcal{C}\cap\mathbb{F}_{p}^n)^\perp = \tr(\mathcal{C}^\perp)$$ where $\tr(x)\colon \mathbb{F}_{p^m}\rightarrow\mathbb{F}_{p}$ sending $x\mapsto x+x^p+...+x^{p^{m-1}}$ is applied componentwise to $\mathcal{C}^\perp$.\end{theorem}
                \vskip 6pt

Since $\GRS$ codes are closed under duality, in the context of $\GRS$ codes Delsarte's theorem can be interpreted as in the following diagram:

$$
\xymatrix@C=4pc@R=2pc{\mathcal{C}=\GRS_k\langle g'\rangle\ar@{<->}[r]_{\textnormal{dual}}\ar[d] & \mathcal{C}^\perp=\GRS_{n-k}\langle g\rangle\ar[d]^{\tr} \\ \mathcal{C}\cap\mathbb{F}_{p}^n=\GRS_k\langle g'\rangle\cap \mathbb{F}_{p}^n \ar@{<->}[r]_{\textnormal{dual}} & \tr(\mathcal{C}^\perp)=\tr(\GRS_{n-k}\langle g\rangle)}
$$

From now on we will focus on subfield-subcodes of $\GRS_{k}\langle g'\rangle$.
A formula for the dimension of the intersection with  $ \mathcal{C}\cap\F_p^n$ is given by
\begin{equation}\label{formula1} \dim_{\F_p}(\mathcal{C}\cap \F_p^n) = n - m(n-k) + \dim_{\F_p}( \ker \tr|_{\mathcal{C}^\perp}).\end{equation}
To obtain a subfield-subcode with large dimension, we need to find an appropriate $g$ so that the trace map restricted to $\{\Ev_{Z_g^c, \bar{1}}(f\cdot g)\colon \deg f < n-k\}$ has large kernel. In the next section, we focus on a particular class of twist polynomials from which we can construct good codes.

To characterize the codes, we need the following definitions.

\begin{definition}\label{d:punccode} Let $\mathcal{C}$ be a code and $S\subset \{1,..., n\}$. The code $\mathcal{C}^S$ obtained from $\mathcal{C}$ by removing the columns of the generator matrix for $\mathcal{C}$ which correspond to the elements of $S$ is called the punctured code (in the coordinates of $S$). \end{definition}
                \vskip 6pt

The punctured code in the coordinates of $S$ is a projection of the vector space $\mathcal{C}$. The outcome of a puncturing is summarized in the following lemma.

\begin{lemma}\label{l:puncparameters} Let $\mathcal{C}$ be an $[n,k,d]$ code and let $\mathcal{C}^i$ be the punctured code in the $i$th coordinate.
\begin{enumerate}
    \item If $d > 1$, $\mathcal{C}^i$ is an $[n-1,k,d^*]$ code where $d^* = d-1$ if $\mathcal{C}$ has a minimum weight codeword with a non-zero $i$th coordinate and $d^* = d$ otherwise.
    \item When $d = 1$, $\mathcal{C}^i$ is an $[n-1,k,1]$ code if $\mathcal{C}$ has no codeword of weight $1$ whose non-zero entry is in coordinate $i$; otherwise, if $k > 1$, $\mathcal{C}^*$ is an $[n-1,k-1,d^*]$ code with $d^* \geq 1$.
\end{enumerate}
\end{lemma}
                \vskip 6pt

\begin{definition}\label{d:shortencode} Let $\mathcal{C}$ be a code, and $S\subset\{1,...,n\}$. Let $\mathcal{C}(S)$ be the set of codewords of $\mathcal{C}$ which are zero in the coordinates of $S$. The shortened code, $\mathcal{C}_S$ is given by puncturing $\mathcal{C}(S)$ in the coordinates of $S$, that is $$\mathcal{C}_S = \mathcal{C}(S)^S.$$
\end{definition}
                \vskip 6pt

A shortened code in the coordinates of $S$ is the intersection of $\mathcal{C}$ with a subspace of $\F_{p^m}^n$, followed by removal of the coordinates of $S$ by puncturing in those positions. The following theorem shows the relationship between puncturing and shortening.

\begin{theorem}{\cite[theorem 1.5.7]{HuffmanPless}}\label{t:puncshortreln} Let $\mathcal{C}$ be a code over $\F_{p^m}$ of length $n$ and let $S$ be a set of $t<n$ coordinates. Denote by $\mathcal{C}^S$ the punctured code in the coordinates of $S$ and $\mathcal{C}_S$ the shortened code in the coordinates of $S$. Then $$(\mathcal{C}^\perp)_S = (\mathcal{C}^S)^\perp.$$
\end{theorem}
                \vskip 6pt

\section{Construction of Punctured Reed-Solomon Codes}

As in the previous section, we set $N = p^m - 1$, $R = \F_{p^m}[x]/\langle x^N-1\rangle$ and we  let $n<N$ be the length of the code. Delsarte's theorem says that to find the dimension of a $\GRS_k\langle g'\rangle\cap \F_p^n$ code, one can find the dimension of the kernel of the trace restricted to the dual code, which is a $\GRS_{n-k}\langle g\rangle$ code, for some $g\in R$. This amounts to finding the dimension of the space $\{f\in R\colon \deg f < n-k,\; \tr(\Ev_{Z_g^c, \bar{1}}(fg)) = 0\}$ over $\F_p$.

\begin{proposition}\label{p:Tmap} Let $f\in R$ and $\T\colon R\rightarrow R$ be the map given by $\T(f) = f + f^p + ... + f^{p^{m-1}}$. Then,
\begin{enumerate}
    \item For $a\in \F_p$, $\T(af) = a\T(f)$,
    \item For every $f\in R$, $\T(f)^p = \T(f^p) = \T(f)$,
    \item For every $f\in R$, $\Ev(\T(f)) = \tr(\Ev(f))$,
    \item $\Ev(\T(f)) = 0$ iff $\T(f) = 0$.
\end{enumerate}
\end{proposition}

\begin{proof} (1), (2), and (3) follow immediately from the definition of $\T$, properties of $\F_p^m$, and the fact that we are working modulo $x^N-1$. (4) follows from the fact that $\Ev$ is an isomorphism. \end{proof}
                \vskip 6pt

From Proposition \ref{p:Tmap}, we conclude that for $g\in R$, \begin{equation}\label{e:space}
\begin{split}
\{f\in R\colon \deg f < n-k,\; \tr(\Ev_{Z_g^c, \bar{1}}(fg)) = 0\} &= \{f\in R\colon \deg f < n-k,\; \Ev_{Z_g^c, \bar{1}}(\T(fg)) = 0\} \\
&= \{f\in R \colon \deg f < n-k,\; \T(fg) = 0\}
\end{split}
\end{equation}

We now restrict ourselves to a special class of functions for $g$, those which satisfy $g^p = g$.

\begin{proposition}\label{p:gprops} Let $g\in R$. Then, the following are equivalent,
\begin{enumerate}
    \item $g = \T(h)$ for some $h\in R$,
    \item $g^p = g$,
    \item $g$ evaluates to $\F_p$.
\end{enumerate}
\end{proposition}

\begin{proof} Suppose that for some $h\in R$, $g = \T(h)$. Then $$g^p = \T(h)^p = T(h) = g$$ where the second equality follows from Proposition \ref{p:Tmap}. If $g^p = g$, then for any $\alpha\in\F_{p^m}$, $g(\alpha)^p = g(\alpha)$ and so $g(\alpha)\in \F_p$. Lastly, suppose that for every $\alpha_i\in \F_{p^m}^*$, $g(\alpha_i)\in \F_p$. Since $\tr$ is surjective, there exists $\beta_i\in\F_{p^m}$ such that $\tr(\beta_i) = g(\alpha_i)$. Let $h$ be the interpolation polynomial satisfying $h(\alpha_i) =  \beta_i.$ Then, $$\Ev (\T(h)) = \tr (\Ev(h)) = \Ev (g)$$ and since $\Ev$ is an isomorphism, $\T(h) = g$.
\end{proof}
                \vskip 6pt

\begin{definition}\label{d:cycdef} Let $g\in R$ be such that $g$ satisfies one of the equivalent conditions in \ref{p:gprops}. Then, we will call $g$ cyclotomic. \end{definition}
                \vskip 6pt

\begin{proposition}\label{p:glinear} Suppose that $g\in R$ is cyclotomic. Then, $\T(fg) = g\T(f)$. \end{proposition}

\begin{proof} We have \begin{align*} \T(fg) &= \sum_{i=0}^{m-1}f^{p^i}g^{p^i} \\&= \sum_{i=0}^{m-1}f^{p^i}g \\ &= g\T(f). \end{align*} \end{proof}
                \vskip 6pt

If $g$ is cyclotomic, then \eqref{e:space} becomes $$\{f\in R\colon \deg f < n-k,\; \T(fg)) = 0\} = \{f\in R\colon \deg f < n-k,\; g\T(f)) = 0\}.$$

\begin{proposition}\label{p:equivcodes} Let $g_1$ and $g_2$ be cyclotomic polynomials with the same roots in $\F_{p^m}$, let $n = |Z_{g_1}^c|$ and $k\leq n$. Then, there is an $\F_p$-monomial transformation taking $\GRS_{n-k}\langle g_1\rangle$ to $\GRS_{n-k}\langle g_2\rangle$.\end{proposition}

\begin{proof} Let $\alpha_1,\alpha_2, \dots,\alpha_{n}$ be the elements of $Z_{g_1}^c=Z_{g_2}^c$.  
 Each codeword of $\GRS_k\langle g_1\rangle$ has the form 
\[(g_1(\alpha_1)f(\alpha_1), ... ,g_1(\alpha_{n})f(\alpha_{n}))\]
 for some $f \in R$ with $\deg(f)< n-k$.
Define the map $\varphi\colon \F^n_{p^m} \rightarrow \F^n_{p^m}$ by 
 $$(c_1, ... c_{n}) \mapsto (c_1 g_{1}^{-1}(\alpha_1) g_2(\alpha_1), ..., c_{n-1}g_{1}^{-1}(\alpha_{n})g_2(\alpha_{n})).$$
Then $\phi$ gives  an isomorphism between $  \GRS_{n-k}\langle g_1\rangle$ and $\GRS_{n-k}\langle g_2\rangle$ for any $k<n$.
\end{proof}
                \vskip 6pt

Proposition \ref{p:equivcodes} shows that when considering Generalized Reed-Solomon codes generated by cyclotomic polynomials, or indeed subfield-subcodes of these codes,  we may restrict attention to polynomials $g$ that evaluate to either 0 or 1 on elements of $\F_{p^m}^*$.  
That is, if $g$ is any cyclotomic polynomial with roots $Z\subset\F_{p^m}^*$, define the polynomial $\hat{g}$ by the interpolation $$\hat{g}(x) = \left\{\begin{array}{ll} 0 & x\in Z_g \\ 1 & x\in \F_{p^m}^*\setminus Z_g. \end{array}\right.$$
Since there is a monomial transformation between $\GRS_{n-k}\langle g\rangle \cap \F_p^n$ and $\GRS_{n-k}\langle \hat{g}\rangle \cap \F_p^n$ the parameters of the two codes are the same.

We also  observe that, for $g$ evaluating to $0$ or $1$, $\GRS_{n-k}\langle g\rangle$ is simply the  punctured Reed-Solomon code of dimension $k$ punctured at $Z_g$.
Conversely, for any punctured Reed-Solomon code there is a polynomial $g\in R$ vanishing on the punctured locations that evaluates to 1 on the remaining locations, and this polynomial is cyclotomic.

\begin{proposition}\label{punccode} There is a one-to-one correspondence between punctured Reed-Solomon codes and Generalized Reed-Solomon codes $\GRS\langle g\rangle$ with $g$ evaluating to $0$ or 1 on all $\alpha \in \F_{p^m}^*$. \end{proposition}
                \vskip 6pt

\section{The Dimension of  Subfield-Subcodes of
   Certain Shortened GRS Codes}
   
 In this section, we derive a lower bound for the  dimension of the dual  to codes $\tr(\GRS_{n-k}\langle g\rangle)$ where  $g$ evaluates to $0$ or $1$ at each $\alpha \in \F_{p^m}$. From the previous section we know that the parameters of such codes cover all possible parameters for  $g$  cyclotomic.  Furthermore, from the duality between puncturing and shortening, these codes are $\GRS_k\langle x\rangle_{Z_g}$.    We analyze the kernel of the map $T$ introduced in the previous section in order to apply Delsarte's theorem and \eqref{formula1}. 


\begin{definition} A cyclotomic coset of $\mathbb{Z}_{N}$ is a subset $I$ such that $I = Ip = \{x\cdot p\colon x\in I\}$. If there is a $t\in \mathbb{Z}_N$ such that every element of $I$ can be written as $t\cdot p^i$ for some $i$, then we call $I$ a minimal cyclotomic coset. For a minimal cyclotomic coset $I$ with smallest element $b\in\mathbb{Z}_N$,  we let $I = I_b$ and $n_b= \lvert I_b\rvert$. The set of smallest elements of the minimal cyclotomic cosets will be denoted by $\mathcal{B}$. \end{definition}
                \vskip 6pt

For every $b\in\mathcal{B}$, $n_b$ is a divisor of $m$.  It is clear that the minimal cyclotomic cosets partition $\mathbb{Z}_N$ and every cyclotomic coset is the union of minimal cyclotomic cosets. 
Consequently, any element $f\in R$ may be decomposed as a sum of polynomials with support in the minimal cyclotomic cosets, so  $f = \sum_{b\in\mathcal{B}}f_b$ and $\textnormal{supp}(f_b)\subseteq I_b$.
Furthermore, $\textnormal{supp} (T(f_b)) \subseteq I_b$.

\begin{proposition} Let $f\in R$ and  $f = \sum_{b\in\mathcal{B}} f_b$ where $\textnormal{supp}f_b\subseteq I_b$. Then $f$ is cyclotomic if and only if for each $b\in\mathcal{B}$, $f_b$ is cyclotomic. Furthermore, for each $f_b$ there exists $\alpha\in \F_{p^m}$ such that $$f_b = \T(\alpha x^b) = \sum_{i=0}^{m-1}(\alpha x^b)^{p^i}.$$
\end{proposition}

\begin{proof} Suppose that $f$ is cyclotomic so that $f = \T(h)$ for some $h\in R$. Writing $h = \sum_{b\in\mathcal{B}} h_b$ we have $\T(h) = \sum_{b\in\mathcal{B}}\T(h_b).$ Letting $f_b = \T(h_b)$ shows that $f$ is the sum of cyclotomic polynomials. The converse is immediate. 

Now, consider a cyclotomic polynomial with support in a minimal cyclotomic coset $I_b$. Say, $f_b = \T(h)$ for some $h\in R$. Then, \begin{align*} f_b &= \T(h) \\ &= \T\left(\sum_{b'\in\mathcal{B}}h_{b'}\right) \\ &= \sum_{b'\in\mathcal{B}}\T(h_{b'}).\end{align*} Since $\T(h_{b'})$ has support in $I_{b'}$, then all terms in the sum must equal zero except for the term $\T(h_b)$. Therefore, we can write \begin{align*} f_b &= \T(h_b) \\ &= \sum_{i=0}^{n_b - 1}\T(\alpha_i x^{bp^i}) \\ &= \sum_{j=0}^{m - 1}\left(\sum_{i=0}^{n_b-1}\alpha_{i}x^{bp^i}\right)^{p^j} \\ &= \sum_{j=0}^{m - 1}x^b\left(\sum_{i=0}^{n_b-1}\alpha_{i}^{p^{m+i}}\right)^{p^{j-i}} \end{align*} where we use that $x^{bp^{n_b}} = x^{bp^m} = x^{b}.$
\end{proof}
                \vskip 6pt

If $h = \sum h_b$, then since each $h_b$ has disjoint support and $\T(h_b)$ has support in $I_b$, we have that $\T(h)=0$ if and only $\T(h_b)=0$ for each $b\in\mathcal{B}$.

\begin{proposition} Let $\eta\in\F_{p^m}$ be a primitive element, and define $g_{b,k,\ell} = \eta^kx^b - \eta^{kp^\ell}x^{bp^\ell}.$ Furthermore, let $\gamma_0, ..., \gamma_{m-n_b-1}$ be an $\F_p$ basis for the kernel of $\tr_{\F_{p^m}:\F_{p^{n_b}}}$. Then, an $\F_p$ basis for $\mathcal{F}_b = \{f\in R \colon \textnormal{supp}f\subseteq I_b, \T(f) = 0\}$ is given by $$\{g_{b,k,\ell}\colon 0\leq k < m,\; 0 < \ell \leq n_b - 1\} \cup \{\gamma_i x^b\colon 0 \leq i < m-n_b\}.$$ \end{proposition}

\begin{proof} For each $b$, it is clear that the polynomials in the set $\{g_{b,k,\ell}\colon 0\leq k < m,\; 0 < \ell \leq n_b - 1\} \cup \{\gamma_i x^b\colon 0 \leq i < m-n_b\}$ must be linearly independent by consideration of degrees. To show that this set is spanning, we recall that the dimension of $\mathcal{F} = \{f\in R \colon T(f) = 0\}$ is $(m-1)N$. For $b\neq b'$, the basis functions for the polynomials with support in $I_b$ and $I_{b'}$ are linearly independent because their supports are disjoint. For a particular $b\in\mathcal{B}$ we can count the number of such functions to be $m (n_b - 1) + m-n_b$. Summing over all of $\mathcal{B}$, we can verify $$\sum_{b\in \mathcal{B}} m(n_b-1)+m-n_b = n_b (m-1) = (m-1)\sum_{b\in\mathcal{B}}n_b = (m-1)N.$$ Since $\mathcal{F}$ is the disjoint union of the $\mathcal{F}_b$, then polynomials above must be a spanning set.
\end{proof}

To obtain a lower bound for $\dim_{\F_p}\left(\GRS_{n-k}\langle x\rangle_{Z_g}\cap\F_p^n\right)$, we can count the polynomials of degree less than $k$ in our basis. This yields the following formula.

\begin{theorem}\label{Mainbound}  Let $g \in R$ be a cyclotomic
  polynomial. Then, $$\dim_{\F_p}\left(\GRS_k\langle
    x\rangle_{Z_g}\cap\F_p^n\right) \geq n-m(n-k)+\sum_{b\in B\cap
    A}\left(m(|I_b\cap A|-1) + m-n_b\right)$$ where $A = \{0,1,...,
  n-k-1\}$.  \end{theorem}

The inequality of the theorem is strict  when there is some $f\in R$ of
degree less than $n-k$ such that $gT(f)= 0$ but $T(f)\ne 0$.

\section{Computation}

We present two algorithms that offer  different search patterns through the same set of codes. The code was
 written in Magma \cite{ma} and compares the results of the
 constructed code with the table of best known codes, given by
 \cite{cota}.   

Algorithm 1 uses cyclotomic polynomials with coefficients in $\F_p$ for $g$ and Algorithm 2 uses
an evaluating set of points $S=\mathbb{F}_{p^m}^{\ast}\setminus
\{\alpha^i: i\in I\}$ where $I$ is a cyclotomic coset. 
It is impractical to test all cases so for  Algorithm~1  we  only use polynomials $g$ whose
coefficients are~1 on a union of  either one, two or three minimal
cyclotomic cosets, and~0 elsewhere.   For Algorithm~2, we used
cyclotomic polynomials vanishing on a union of one, two or three
minimal cyclotomic cosets.
 It turns out that Algorithm 1 produced many more new codes. 

For the codes tested, the lower bound for the dimension in  Theorem
\ref{Mainbound} was  never larger than, although it was often equal
to,  the dimension of codes in the tables of best known codes.
In order to beat the best known codes we used Magma to compute the
true dimension, which can exceed the bound.  Thus it would be
interesting to improve our lower bound by investigating those  $f\in R$
such that $T(fg)=0$ although $T(f)\ne 0$.


\begin{algorithm}[htb]
  \caption{Search for new codes with cyclotomic polynomials}\label{alg1}
 \pagebreak
  \begin{algorithmic}[1]
    \REQUIRE A finite field $\mathbb{F}_{p^m}$ with primitive root $\alpha$. \\
    \ENSURE List L with new codes over $\mathbb{F}_p$.\\
    \medskip
   BEGIN
   \STATE Find a set of cyclotomic cosets, $I = \{I_{b_i} : I_{b_i} \subset \{1, ..., p^r-1\}\}$
   \STATE For each $I_b\in I$ consider $g=\sum_{i\in I_b} x^i$.
   \medskip
   \STATE Consider $S:=[\alpha^i: g(\alpha^i)\neq 0]$.
   \STATE Consider $v:=[Evaluate(g,\alpha^i): \alpha^i \in S]$.
    \STATE FOR $k=1,\ldots,\# S$ DO
 \STATE $C=\GRS_k(S,v)$; $D=TraceCode(C)$; $E=Dual(D)$.
    \STATE IF  $E$  is a new code with respect to \cite{cota} include it in L ENDIF.
\STATE ENDFOR
\STATE Return(L).

 END\\
 \end{algorithmic}
\end{algorithm}

\begin{algorithm}[htb]
    \caption{Search for new codes with cyclotomic roots}\label{alg2}
    \pagebreak
    \begin{algorithmic}
        \REQUIRE A finite field $\mathbb{F}_{p^m}$ \\
        \ENSURE List L with new codes over $\F_p$ \\
        \medskip
        BEGIN
         \STATE Find a set of cyclotomic cosets, $I = \{I_{b_i} : I_{b_i} \subset \{1, ..., p^r-1\}\}$
   \STATE For each $I_b\in I$ consider  $S=\mathbb{F}_{p^m}^{\ast}\setminus \{\alpha^i: i\in I_b\}$        
        \STATE Set $v = (1,..., 1)$
        \STATE FOR $k = 1$ to $n$ DO
        \STATE $C = \GRS_k(S, v)$; $D = TraceCode(C)$; $E=Dual(D)$.
        \STATE If $E$ is a new code with respect to \cite{cota} include it in L ENDIF
        \STATE ENDFOR
        \STATE Return(L).
        
    END\\
    \end{algorithmic}
\end{algorithm}

In the following subsections we display in different tables the new codes we have obtained over $\mathbb{F}_2$, $\mathbb{F}_3$ and $\mathbb{F}_5$.
The first column corresponds with the parameters of the best known linear codes \cite{cota}.
In the second column we write the parameters of the new codes. The minimum distance of the new codes is not directly computed, rather we use the lower bound given by the minimum distance of the original code.  The dimension is computed by constructing the dual code---applying the function $T$ to a generating set for the parent code over $\F_{p^m}$---and then using linear algebra to compute the dimension of the dual code.  We only list codes that are either better than the previously best-known codes, or  equal to best-known codes for which there was only an existence proof and no known construction.

The last column gives either the polynomial $g$ we  used in our algorithm or the operation used to get
new codes from other codes, i.e., either ShortenCode or PunctureCode.

The latter operations are used to obtain new codes from existing ones.  If $C$ is a linear code with parameters $[n,k,d]$
then $ShortenCode(C)$ is a linear code with parameters $[n-1,\geq k-1,d]$ and $PunctureCode(C)$ is a linear code with parameters
$[n-1,k,\geq d-1]$.

\subsection{New codes over $\mathbb{F}_2$}

We have obtained two new codes from Algorithm \ref{alg1} in Fig \ref{fig:1}, and we can get more new codes from $C_1$ and $C_2$ see Fig \ref{fig:2}.

\begin{figure*}[t]\caption{New Codes over $\mathbb{F}_2$ obtained with Algorithm \ref{alg1}}
\begin{center}
\begin{tabular}{|c|c|c|}\hline
Best Known  & New codes &  $g$   \\ \hline

$[192,66,39]$ & $C_1=[192,66,40] $ &   $x^{192} + x^{144} + x^{132} + x^{129} + x^{96} + x^{72} + x^{66} + x^{48} + x^{36}$ \\
&& $ + x^{33} + x^{24} + x^{18} + x^{12} + x^9 + x^6 + x^3$ \\ \hline\hline
$[240,76,51]$ & $C_2=[240,76,52]$ & $x^{240} + x^{225} + x^{210} + x^{195} + x^{180} + x^{165} + x^{150}$\\
&&$ + x^{135} + x^{120} + x^{105} + x^{90} + x^{75} + x^{60} + x^{45} + x^{30} + x^{15}$ \\ \hline
\end{tabular}
\end {center}
\label{fig:1}
\end{figure*}

\begin{figure*}[t]\caption{New Codes over $\mathbb{F}_2$ obtained from $C_1$ and $C_2$}
\begin{center}
\begin{tabular}{|c|c|c|}\hline
Best Known  & New codes &  Operation   \\ \hline
$[191,66,38]$ & $C_3=[191,66,39]$ &  PunctureCode($C_1,192$)\\ \hline
$[191,65,39]$ & $C_4=[191,65,40] $ &  ShortenCode($C_1,192$)\\ \hline
$[190,65,38]$ & $C_5=[190,65,39] $ &  ShortenCode($C_3,191$)\\ \hline \hline

$[239,76,50]$ & $C_6=[239,76,51] $ &   PunctureCode($C_2,240$)\\ \hline
$[238,76,49]$ & $C_7=[238,76,50] $ &   PunctureCode($C_2,\{240,239\}$)\\ \hline
$[237,76,48]$ & $C_8=[237,76,49] $ &   PunctureCode($C_2,\{240,239,238\}$)\\ \hline

$[239,75,51]$ & $C_9=[239,75,52] $ &   ShortenCode($C_2,240$)\\ \hline
$[238,74,51]$ & $C_{10}=[238,74,52] $ &   ShortenCode($C_2\{240,239\}$)\\ \hline
$[237,73,51]$ & $C_{11}=[237,73,52] $ &   ShortenCode($C_2,\{240,239,238\}$)\\ \hline

$[238,75,50]$ & $C_{12}=[238,75,51] $ &   ShortenCode($C_6,239$)\\ \hline
$[237,74,50]$ & $C_{13}=[237,74,51] $ &   ShortenCode($C_6\{239,238\}$)\\ \hline
$[236,73,50]$ & $C_{14}=[236,73,51] $ &   ShortenCode($C_6,\{239,238,237\}$)\\ \hline

$[237,75,49]$ & $C_{14}=[237,75,50] $ &   ShortenCode($C_7,238$)\\ \hline
$[236,74,49]$ & $C_{15}=[236,74,50] $ &   ShortenCode($C_7\{238,237\}$)\\ \hline
$[235,73,49]$ & $C_{16}=[235,73,50] $ &   ShortenCode($C_7,\{238,237,236\}$)\\ \hline

$[236,75,48]$ & $C_{17}=[236,75,49] $ &   ShortenCode($C_8,237$)\\ \hline
$[235,74,48]$ & $C_{18}=[235,74,49] $ &   ShortenCode($C_8\{237,236\}$)\\ \hline
$[234,73,48]$ & $C_{19}=[234,73,49] $ &   ShortenCode($C_8,\{237,236,235\}$)\\ \hline
\end{tabular}
\end {center}
\label{fig:2}
\end{figure*}

\subsection{New codes over $\mathbb{F}_3$}
We have obtained $14$ new codes from Algorithm \ref{alg1} in Fig \ref{fig:3}, and we can get more new codes from them see Fig \ref{fig:4}.

\begin{figure*}[t]\caption{New Codes over $\mathbb{F}_3$ obtained with Algorithm \ref{alg1}}
\begin{center}
\begin{tabular}{|c|c|c|}\hline
Best Known  & New codes &  $g$   \\ \hline

$[162,92,23]$ & $C_{20}=[162,92,23] $ &   $x^{81} + x^{27} + x^9 + x^3 + x$ \\ \hline
$[162,97,21]$ & $C_{21}=[162,97,21] $ &   $x^{81} + x^{27} + x^9 + x^3 + x$ \\ \hline
$[162,102,19]$ & $C_{22}=[162,102,20] $ &   $x^{81} + x^{27} + x^9 + x^3 + x$ \\ \hline
$[162,107,17]$ & $C_{23}=[162,107,18] $ &   $x^{81} + x^{27} + x^9 + x^3 + x$ \\ \hline
$[162,117,14]$ & $C_{24}=[162,117,15] $ &   $x^{81} + x^{27} + x^9 + x^3 + x$ \\ \hline\hline

$[161,91,23]$ & $C_{25}=[161,91,23] $ &   $x^{202} + x^{148} + x^{130} + x^{124} + x^{122} + x^{121}$ \\ \hline
$[161,96,21]$ & $C_{26}=[161,96,21] $ &   $x^{202} + x^{148} + x^{130} + x^{124} + x^{122} + x^{121}$ \\ \hline
$[161,101,19]$ & $C_{27}=[161,101,20] $ &   $x^{202} + x^{148} + x^{130} + x^{124} + x^{122} + x^{121}$ \\ \hline
$[161,106,17]$ & $C_{28}=[161,106,18] $ &   $x^{202} + x^{148} + x^{130} + x^{124} + x^{122} + x^{121}$ \\ \hline
$[161,116,14]$ & $C_{29}=[161,116,15] $ &   $x^{202} + x^{148} + x^{130} + x^{124} + x^{122} + x^{121}$ \\ \hline\hline

$[171,101,23]$ & $C_{30}=[171,101,23] $ &   $x^{162} + x^{81} + x^{54} + x^{27} + x^{18} + x^9 + x^6 + x^3 + x^2 + x$ \\ \hline
$[171,81,32]$ & $C_{31}=[171,81,32] $ &   $x^{162} + x^{81} + x^{54} + x^{27} + x^{18} + x^9 + x^6 + x^3 + x^2 + x$ \\ \hline\hline

$[170,100,23]$ & $C_{32}=[170,100,23] $ &   $x^{175} + x^{139} + x^{127} + x^{123} + x^{121} + x^{41}$ \\ \hline
$[170,80,32]$ & $C_{33}=[170,80,32] $ &   $x^{175} + x^{139} + x^{127} + x^{123} + x^{121} + x^{41}$ \\ \hline
\end{tabular}
\end {center}
\label{fig:3}
\end{figure*}

\begin{figure*}[t]\caption{New Codes over $\mathbb{F}_3$ obtained $C_{20},\ldots,C_{33}$}
\begin{center}
\begin{tabular}{|c|c|c|}\hline
Best Known  & New codes &  operation   \\ \hline
$[160,95,21]$ & $C_{34}=[160,95,21] $ & ShortenCode($C_{26},161$)   \\ \hline
$[159,94,21]$ & $C_{35}=[159,94,21] $ & ShortenCode($C_{26},\{161,160\}$)   \\ \hline
$[158,93,21]$ & $C_{36}=[158,93,21] $ & ShortenCode($C_{26},\{161,160,159\}$)   \\ \hline\hline

$[160,100,19]$ & $C_{37}=[160,100,20] $ &   ShortenCode($C_{27},161$)  \\ \hline
$[159,99,19]$ & $C_{38}=[159,99,20] $ &  ShortenCode($C_{27},\{161,160\}$)\\ \hline

$[161,102,19]$ & $C_{39}=[161,102,19] $ &   PunctureCode($C_{27},162$)  \\ \hline
$[160,101,19]$ & $C_{40}=[160,101,19] $ &   PunctureCode($C_{18},\{161\}$)\\ \hline
$[159,100,19]$ & $C_{41}=[159,100,19] $ &   ShortenCode($C_{40},160$)  \\ \hline
$[158,99,19]$ & $C_{42}=[158,99,19] $ &  ShortenCode($C_{40},\{159\}$)\\ \hline\hline

$[160,105,17]$ & $C_{43}=[160,105,18] $ &  ShortenCode($C_{28},161$)\\ \hline
$[159,104,17]$ & $C_{44}=[159,104,18] $ &  ShortenCode($C_{28},\{161,160\}$)\\ \hline
$\vdots$ & $\vdots$ &  $\vdots$\\ \hline
$[134,180,18]$ & $C_{68}=[135,180,18] $ &  ShortenCode($C_{28},\{161,\ldots,136\}$)\\ \hline\hline

$[160,115,14]$ & $C_{69}=[160,115,15] $ &  ShortenCode($C_{29},161$)\\ \hline
$\vdots$ & $\vdots$ &  $\vdots$\\ \hline
$[148,103,14]$ & $C_{80}=[148,103,15] $ &  ShortenCode($C_{29},\{161,\ldots,149\}$)\\ \hline\hline

$[169,99,23]$ & $C_{81}=[169,99,23] $ &   ShortenCode($C_{32},\{170\}$)\\ \hline
$\vdots$ & $\vdots$ &  $\vdots$\\ \hline
$[161,91,23]$ & $C_{89}=[161,91,23] $ &   ShortenCode($C_{32},\{170,\ldots,162\}$)\\ \hline\hline

$[169,79,32]$ & $C_{90}=[169,79,32] $ &   ShortenCode($C_{33},\{170\}$)\\ \hline
$[168,78,32]$ & $C_{91}=[168,78,32] $ &   ShortenCode($C_{33},\{170,169\}$)\\ \hline
$[167,77,32]$ & $C_{92}=[167,77,32] $ &   ShortenCode($C_{33},\{170,169,168\}$)\\ \hline

\end{tabular}
\end {center}
\label{fig:4}
\end{figure*}

\subsection{New codes over $\mathbb{F}_5$}
We have obtained $3$ new codes from Algorithm \ref{alg1} in Fig \ref{fig:5}, and we can get more new codes from them see Fig \ref{fig:6}.

\begin{figure*}[t]\caption{New Codes over $\mathbb{F}_5$ obtained with Algorithm \ref{alg1}}
\begin{center}
\begin{tabular}{|c|c|c|}\hline
Best Known  & New codes &  $g$   \\ \hline

$[100,33,35]$ & $C_{93}=[100,33,35] $ &   $x^{25} + x^5 + x$ \\ \hline
$[100,36,33]$ & $C_{94}=[100,36,34] $ &   $x^{25} + x^5 + x$ \\ \hline\hline
$[99,35,33]$ & $C_{95}=[99,35,34] $ &   $x^{56} + x^{36} + x^{32} + x^{31}$ \\ \hline
\end{tabular}
\end{center}
\label{fig:5}
\end{figure*}

\begin{figure*}[t]\caption{New Codes over $\mathbb{F}_5$ obtained $C_{95}$}
\begin{center}
\begin{tabular}{|c|c|c|}\hline
Best Known  & New codes &  Operation   \\ \hline

$[98,34,33]$ & $C_{96}=[98,34,34] $ &   ShortenCode($C_{95},\{99\}$)\\ \hline
$[97,33,33]$ & $C_{97}=[97,33,34] $ &   ShortenCode($C_{95},\{99,98\}$)\\ \hline
$[96,32,34]$ & $C_{98}=[96,32,34] $ &   ShortenCode($C_{95},\{99,98,97\}$)\\ \hline
\end{tabular}
\end{center}
\label{fig:6}
\end{figure*}

\end{document}